\documentclass[onecolumn]{revtex4-2}
\usepackage[T1]{fontenc}
\usepackage{mathtools}
\usepackage{amssymb}
\usepackage{amsthm}
\usepackage{thmtools}
\usepackage{hyperref}
\usepackage{physics}
\usepackage{graphicx}
\usepackage{subfigure}
\usepackage{float}
\usepackage{bm}
\usepackage{stmaryrd}
\usepackage{verbatim}

\newtheorem{definition}{Definition}
\newtheorem{lemma}{Lemma}
\newtheorem{theorem}{Theorem}
\newtheorem{proposition}{Proposition}

\newtheorem{example}{Example}

\begin{document}
	
	\title{Quantifying and detecting quantum-state texture}

	\author{Xiangyu Chen}
	\email{Electronic address: 23S012014@stu.hit.edu.cn}
	\affiliation{School of Mathematics, Harbin Institute of Technology, Harbin 150001, China}
	\author{Qiang Lei}
	\email{Electronic address: leiqiang@hit.edu.cn}
	\affiliation{School of Mathematics, Harbin Institute of Technology, Harbin 150001, China}
	
	\begin{abstract}
		Quantum-state texture is a recently proposed quantum resource that characterizes the inhomogeneity of a quantum state's matrix element distribution in the computational basis, enriching our understanding of quantum state structure. To expand its quantification toolkit and establish detection methods, in this article, we investigate the resource theory of texture from both quantitative and detection perspectives. First, we construct a texture measure $\mathcal{T}^{\text{GR}}_{\alpha,z}(\rho)$ based on the $\alpha$-$z$ R\'enyi relative entropy and present some of its inherent properties. Second, we analyze the mathematical relationships between several existing texture measures, revealing connections among different quantifiers. Finally, drawing on the witness concept from other resource theories, we systematically introduce texture witnesses into the texture theory and provide examples of texture witnesses with special properties.
		\par\noindent
		\textbf{keywords}: state texture, quantum resource theory, $\alpha$-$z$ R\'enyi relative entropy, witness
	\end{abstract}
	\maketitle
	
	\section{Introduction}
	\par\noindent\par 
	In quantum information theory, quantifying the physical properties of quantum systems has long been a focus of research. Quantum resource theory, primarily used to quantify and manage these physical properties, has given rise to resource theories such as entanglement \cite{Horodecki2009Quantum}, coherence \cite{baumgratz2014quantifying}, imaginarity \cite{hickey2018quantifying}, magic \cite{Howard2017application}, among others \cite{Gour2015resource,xu2016quantifying,Luo2017partial,Amaral2018Noncontextual,xu2019coherence,xu2023quantifying}. These theories treat specific physical attributes as operational resources, classifying all quantum states into resource states, which contain the target resource, and free states, which do not. Similarly, physical operations are categorized into resource operations, which may consume or transform the resource, and free operations, which cannot generate it. This formal theoretical framework not only deepens the understanding of the nature of quantum physical properties but also directly or indirectly guides tasks in quantum computing, communication, and cryptography (e.g., imaginarity resources can be utilized in the discrimination of quantum states and channels \cite{wu2021operational,wu2024resource}). Consequently, it enhances performance while optimizing resource consumption, vigorously promoting the practical application of quantum information processing.
	\par
	In 2024, Parisio made foundational work at the frontier of quantum resource theory by first proposing and systematically establishing the resource-theoretic framework for quantum-state texture (QST) \cite{parisio2024quantum}. On a Hilbert space $\mathcal{H}$, let $\mathcal{D}(\mathcal{H})$ denote the set of density matrices. After choosing a computational basis $\{\ket{i}\}$, the matrix elements of a density matrix completely determine its representation. If the real or imaginary part of each matrix element is mapped to a height coordinate in a three-dimensional space, each quantum state corresponds to a specific three-dimensional surface. The surfaces corresponding to the vast majority of quantum states are not entirely flat, exhibiting variations and undulations in their numerical distribution. This intrinsic deviation from a uniform distribution within the mathematical representation of a quantum state is defined as quantum-state texture in resource theory. Since texture is directly related to the distribution of amplitudes and phases of a quantum state under a specific basis, it provides a new dimension for characterizing quantum state structure and shows potential application value in quantum information, quantum computing, and quantum biology \cite{parisio2024quantum}. According to the QST resource theory framework, the unique texture-free state $f_1$ under the basis $\{\ket{j}\}_{j=1}^d$ is defined as:
	\begin{align*}
		f_1 = \ket{f_1}\bra{f_1}, \quad \text{where} \quad \ket{f_1} = \frac{1}{\sqrt{d}} \sum_{j=1}^{d} \ket{j}.
	\end{align*}
	Texture-free operations are all completely positive trace-preserving (CPTP) maps $\Phi: \rho \to \Phi(\rho)$ that satisfy the fixed-point condition: $\Phi(f_1) = f_1$. Let the Kraus operator representation of $\Phi$ be $\Phi(\cdot) = \sum_n K_n \cdot K_n^\dagger$ (satisfying $\sum_n K_n^\dagger K_n = I$). Since $f_1$ is a pure state, the above condition is equivalent to:
	\begin{align*}
		K_n \ket{f_1} = \alpha_n \ket{f_1}, \quad \forall n, \quad \alpha_n \in \mathbb{C}.
	\end{align*}
	\par
	Having clarified the definitions of free states and free operations, the core task of quantification is to define a texture measure $\mathcal{T}(\rho)$ that quantifies the amount of texture resource. A valid texture measure must satisfy the following three axioms \cite{parisio2024quantum}:
	\par\noindent
	\textbf{(T1) Non-negativity:} For any quantum state $\rho$, $\mathcal{T}(\rho) \geqslant 0$, and $\mathcal{T}(f_1) = 0$.
	\par\noindent
	\textbf{(T2) Monotonicity:} For any quantum state $\rho$ and any texture-free operation $\Phi$, $\mathcal{T}(\rho) \geqslant \mathcal{T}(\Phi(\rho))$.
	\par\noindent
	\textbf{(T3) Convexity:} For any set of quantum states $\{\rho_n\}$ and a probability distribution $\{p_n\}$ satisfying $p_n \geqslant 0$, $\sum_n p_n = 1$, $\mathcal{T}\left(\sum_n p_n \rho_n\right) \leqslant \sum_n p_n \mathcal{T}(\rho_n)$.
	\par
	As part of constructing the texture theory, Parisio simultaneously proposed the first concrete and experimentally accessible texture measure, state rugosity \cite{parisio2024quantum}. Subsequently, a series of researchers, following this framework, have proposed diverse texture measures from different mathematical and physical perspectives, advancing quantification research in this field. In late 2024, the work by Wang \textit{et al.} \cite{wang2025quantifying} initiated systematic research on texture measures. They proposed texture measures based on trace distance, geometric measure, fidelity, and Bures distance, rigorously proving they satisfy the aforementioned axioms. Their work not only provided practical measurement tools but also revealed that quantifiers common in other resource theories, such as the $l_1$-norm, robustness, and quantum relative entropy, are not suitable as valid measures in texture theory, highlighting the uniqueness of texture as a resource.
	\par
	In July 2025, the research by Zhang \textit{et al.} \cite{zhang2025quantum} further expanded the scope of measures by innovatively proposing texture measures based on weight and Tsallis relative entropy. Beyond providing new quantification methods, the deeper significance of this work lies in their demonstration that the replacement framework, which holds universally in resource theories like coherence \cite{yu2016alternative}, imaginarity \cite{xue2021quantification}, and block coherence \cite{xu2020general}, does not hold in the texture resource theory. This discovery underscores the distinct structural characteristics of texture theory compared to other resource theories. It is worth noting that the texture measure based on Tsallis relative entropy was also independently proposed and studied in depth by Cui \textit{et al.} in August 2025 \cite{cui2025quantum}. Their work additionally proposed texture measures based on $\alpha$-affinity and the $l_2$-norm, and quantified the interrelationships between texture resource and resources such as coherence, imaginarity, and predictability from multiple angles.
	\par
	In October 2025, Muthuganesan \cite{muthuganesan2026quantum} presented texture measures based on the Hellinger distance, quantum Jensen-Shannon divergence, and Wigner-Yanase-Dyson skew information. Muthuganesan also proved that the texture measure based on the Jensen-Shannon divergence obeys the no-broadcasting theorem \cite{Barnum1996Noncommuting}, providing an operational meaning for the texture resource in information processing tasks. Almost simultaneously, another work by Cao \textit{et al.} in October 2025 \cite{cao2026generalized} independently explored the measure based on the quantum Jensen-Shannon divergence, emphasizing its advantage as a well-defined, bounded alternative that avoids the divergence issues sometimes encountered with relative entropy measures. Furthermore, Cao \textit{et al.} extended the measure system to more general relative entropy forms, proposing texture measures based on sandwiched R\'enyi relative entropy and unified $(\alpha, \beta)$ relative entropy. These diverse texture measures demonstrate the variety and hierarchy of tools for quantifying texture.
	\par
	In this article, we investigate the resource theory of texture from both quantitative and detection dimensions. The structure of our work is as follows. In Section II, we introduce a new texture measure and analyze its properties. In Section III, we review several known valid texture measures and establish mathematical relationships between them. In Section IV, inspired by other resource theories, we define texture witnesses within the texture resource theory and provide construction methods for some texture witnesses with specific properties. Finally, a concise summary is presented in Section V.
	
	\section{A New Texture Measure}
	\par\noindent\par  
	To enrich the toolkit for quantifying texture, we construct a new texture measure based on the $\alpha$-$z$ Rényi relative entropy. This two-parameter functional is a generalization of various relative entropies, and the $\alpha$-$z$ Rényi relative entropy possesses the formal properties (e.g., data processing inequality) required for a valid resource measure. Consequently, a texture measure derived from it provides a unifying framework that naturally encompasses and connects several existing measures. Let $\mathcal{P}(\mathcal{H})$ denote the set of positive semidefinite matrices. For $\tau,\sigma\in\mathcal{D}(\mathcal{H})$, the $\alpha$-$z$ Rényi relative entropy is defined as \cite{audenaert2015alpha}:
	\begin{align*}
		D_{\alpha,z}(\tau||\sigma)=\frac{1}{\alpha-1}\log f_{\alpha,z}(\tau||\sigma),
	\end{align*}
	where $f_{\alpha,z}(\tau||\sigma)=\operatorname{Tr}(\tau^{\frac{\alpha}{2z}}\sigma^{\frac{1-\alpha}{z}}\tau^{\frac{\alpha}{2z}})^z=\operatorname{Tr}(\sigma^{\frac{1-\alpha}{2z}}\tau^{\frac{\alpha}{z}}\sigma^{\frac{1-\alpha}{2z}})^z$.
	\par
	To construct and prove the texture measure to be presented, we first introduce a matrix inequality \cite{audenaert2007araki}:
	\begin{lemma}[Araki-Lieb-Thirring inequality]
		Let $A$ and $B$ be two positive semidefinite matrices, $r \geqslant 1$ and $q \geqslant 0$. The following inequality holds:
		\begin{align*}
			\operatorname{Tr}(ABA)^{rq} \leqslant \operatorname{Tr}(A^r B^r A^r)^q.
		\end{align*}
	\end{lemma}
	\par
	\begin{theorem}
		The function $\mathcal{T}^{\text{GR}}_{\alpha,z}(\rho)$ defined based on the $\alpha$-$z$ R\'enyi relative entropy is a valid texture measure:
		\begin{eqnarray*}
			\mathcal{T}^{\text{GR}}_{\alpha,z}(\rho)=1-f_{\alpha,z}(f_1||\rho),
		\end{eqnarray*}
		where $\alpha \in (0,1)$ and $\max\{\alpha, 1-\alpha\} \leqslant z$.
	\end{theorem}
	\begin{proof}
		\textbf{(T1) Non-negativity:}
		From \cite{muller2013quantum}, we have $D_{\alpha,\alpha}(\tau||\sigma) \geqslant 0$, which implies $\log \operatorname{Tr}(\sigma^{\frac{1-\alpha}{2\alpha}}\tau^{\frac{\alpha}{\alpha}}\sigma^{\frac{1-\alpha}{2\alpha}})^\alpha \leqslant 0$, i.e., $\operatorname{Tr}(\sigma^{\frac{1-\alpha}{2\alpha}}\tau^{\frac{\alpha}{\alpha}}\sigma^{\frac{1-\alpha}{2\alpha}})^\alpha \leqslant 1$. Utilizing Lemma 1, we obtain:
		\begin{align*}
			f_{\alpha,z}(f_1||\rho)=\operatorname{Tr}(\rho^{\frac{1-\alpha}{2z}}f_1^{\frac{\alpha}{z}}\rho^{\frac{1-\alpha}{2z}})^z=\operatorname{Tr}(\rho^{\frac{1-\alpha}{2z}}f_1^{\frac{\alpha}{z}}\rho^{\frac{1-\alpha}{2z}})^{\frac{z}{\alpha}\alpha}\leqslant\operatorname{Tr}(\rho^{\frac{1-\alpha}{2\alpha}}f_1^{\frac{\alpha}{\alpha}}\rho^{\frac{1-\alpha}{2\alpha}})^\alpha\leqslant 1.
		\end{align*}
		Thus, $\mathcal{T}^{\text{GR}}_{\alpha,z}(\rho) \geqslant 0$. It is obvious that $\mathcal{T}^{\text{GR}}_{\alpha,z}(f_1) = 0$.
		\par\noindent
		\textbf{(T2) Monotonicity:} Since the $\alpha$-$z$ R\'enyi relative entropy satisfies the data processing inequality (DPI) \cite{audenaert2015alpha}, i.e., $D_{\alpha,z}(\tau||\sigma) \geqslant D_{\alpha,z}[\Phi(\tau)||\Phi(\sigma)]$, we have the following equivalent chain:
		\begin{align*}
			&\ D_{\alpha,z}(f_1||\rho) \geqslant D_{\alpha,z}[\Phi(f_1)||\Phi(\rho)] = D_{\alpha,z}[f_1||\Phi(\rho)] \\
			\Leftrightarrow &\ \log f_{\alpha,z}(f_1||\rho) \leqslant \log f_{\alpha,z}(f_1||\Phi(\rho)) \\
			\Leftrightarrow &\ f_{\alpha,z}(f_1||\rho) \leqslant f_{\alpha,z}(f_1||\Phi(\rho)).
		\end{align*}
		Therefore, $\mathcal{T}^{\text{GR}}_{\alpha,z}(\rho) \geqslant \mathcal{T}^{\text{GR}}_{\alpha,z}[\Phi(\rho)]$.
		\par\noindent
		\textbf{(T3) Convexity:} It suffices to consider the case $n=2$. Since $f_{\alpha,z}$ satisfies joint concavity \cite{audenaert2015alpha}, we have:
		\begin{align*}
			&\ p_1\mathcal{T}^{\text{GR}}_{\alpha,z}(\rho_1) + p_2\mathcal{T}^{\text{GR}}_{\alpha,z}(\rho_2) \\
			= &\ p_1 + p_2 - \big[ p_1 f_{\alpha,z}(f_1||\rho_1) + p_2 f_{\alpha,z}(f_1||\rho_2) \big] \\
			\geqslant &\ 1 - f_{\alpha,z}\big[ p_1 f_1 + p_2 f_1 \ || \ p_1 \rho_1 + p_2 \rho_2 \big] \\
			= &\ 1 - f_{\alpha,z}\big[ f_1 \ || \ p_1 \rho_1 + p_2 \rho_2 \big] \\
			= &\ \mathcal{T}^{\text{GR}}_{\alpha,z}(p_1 \rho_1 + p_2 \rho_2).
		\end{align*}
		Hence, $\mathcal{T}^{\text{GR}}_{\alpha,z}(\rho)$ satisfies convexity.
	\end{proof}
	\par
	Based on this measure, we further investigate its parameter dependence and properties under specific transformations. Concerning the two parameters $\alpha$ and $z$, $\mathcal{T}^{\text{GR}}_{\alpha,z}(\rho)$ possesses the following properties:
	\begin{proposition}
		For any density matrix $\rho$, the following holds:
		\par\noindent
		(1) If $\alpha_1 \leqslant \alpha_2$, then $\mathcal{T}^{\text{GR}}_{\alpha_1,\alpha_1}(\rho) \leqslant \mathcal{T}^{\text{GR}}_{\alpha_2,\alpha_2}(\rho)$;
		\par\noindent
		(2) If $z_1 \leqslant z_2$, then $\mathcal{T}^{\text{GR}}_{\alpha,z_1}(\rho) \leqslant \mathcal{T}^{\text{GR}}_{\alpha,z_2}(\rho)$.
	\end{proposition}
	\begin{proof}
		(1) Since $D_{\alpha,\alpha}(f_1||\rho)$ is a monotonically increasing function of $\alpha$ \cite{muller2013quantum}, and given $\alpha-1 < 0$ along with the monotonicity of the exponential function, it follows that $f_{\alpha_1,\alpha_1}(f_1||\rho) \geqslant f_{\alpha_2,\alpha_2}(f_1||\rho)$. Thus, the conclusion holds.
		\par\noindent
		(2) Since $\frac{z_2}{z_1} \geqslant 1$, applying Lemma 1 yields:
		\begin{align*}
			f_{\alpha,z_2}(f_1||\rho) &= \operatorname{Tr}(f_1^{\frac{\alpha}{2z_2}}\rho^{\frac{1-\alpha}{z_2}}f_1^{\frac{\alpha}{2z_2}})^{z_2} = \operatorname{Tr}(f_1^{\frac{\alpha}{2z_2}}\rho^{\frac{1-\alpha}{z_2}}f_1^{\frac{\alpha}{2z_2}})^{\frac{z_2}{z_1} \cdot z_1} \\
			&\leqslant \operatorname{Tr}(f_1^{\frac{\alpha}{2z_1}}\rho^{\frac{1-\alpha}{z_1}}f_1^{\frac{\alpha}{2z_1}})^{z_1} = f_{\alpha,z_1}(f_1||\rho).
		\end{align*}
		Consequently, $\mathcal{T}^{\text{GR}}_{\alpha,z_1}(\rho) \leqslant \mathcal{T}^{\text{GR}}_{\alpha,z_2}(\rho)$ holds.
	\end{proof}
	Furthermore, from the axiomatic properties of the $\alpha$-$z$ R\'enyi relative entropy, we obtain:
	\begin{proposition}
		In a $d$-dimensional system, if a unitary operation satisfies $U f_1 U^{\dagger} = f_1$, then $\mathcal{T}^{\text{GR}}_{\alpha,z}(U\rho U^{\dagger}) = \mathcal{T}^{\text{GR}}_{\alpha,z}(\rho)$.
	\end{proposition}
	\begin{proof}
		Since the $\alpha$-$z$ R\'enyi relative entropy satisfies unitary invariance \cite{audenaert2015alpha}, we have:
		\begin{align*}
			&\ D_{\alpha,z}(f_1||\rho) = D_{\alpha,z}[U f_1 U^{\dagger} || U\rho U^{\dagger}] = D_{\alpha,z}[f_1 || U\rho U^{\dagger}] 
			\\
			\Leftrightarrow &\ f_{\alpha,z}(f_1||\rho) = f_{\alpha,z}(f_1 || U\rho U^{\dagger}).
		\end{align*}
		Thus, the conclusion follows.
	\end{proof}
	\begin{proposition}
		For any density matrices $\rho \in \mathcal{D}(\mathcal{H}_{1})$ and $\delta \in \mathcal{D}(\mathcal{H}_{2})$, the following inequality holds:
		$\mathcal{T}^{\text{GR}}_{\alpha,z}(\rho) + \mathcal{T}^{\text{GR}}_{\alpha,z}(\delta) \geqslant \mathcal{T}^{\text{GR}}_{\alpha,z}(\rho \otimes \delta) \geqslant \mathcal{T}^{\text{GR}}_{\alpha,z}(\rho) \mathcal{T}^{\text{GR}}_{\alpha,z}(\delta)$.
	\end{proposition}
	\begin{proof}
		Let $f_1^{\rho}$ and $f_1^{\delta}$ denote the free states in $\mathcal{D}(\mathcal{H}_{1})$ and $\mathcal{D}(\mathcal{H}_{2})$, respectively. Then $f_1^{\rho\otimes\delta} := f_1^{\rho} \otimes f_1^{\delta}$ is the free state in the tensor product system. Since the $\alpha$-$z$ R\'enyi relative entropy satisfies additivity \cite{audenaert2015alpha}, we have:
		\begin{align*}
			&\ D_{\alpha,z}(f_1^{\rho\otimes\delta} || \rho \otimes \delta) = D_{\alpha,z}(f_1^{\rho} || \rho) + D_{\alpha,z}(f_1^{\delta} || \delta) \\
			\Leftrightarrow &\ \log f_{\alpha,z}(f_1^{\rho\otimes\delta} || \rho \otimes \delta) = \log f_{\alpha,z}(f_1^{\rho} || \rho) + \log f_{\alpha,z}(f_1^{\delta} || \delta) \\
			\Leftrightarrow &\ f_{\alpha,z}(f_1^{\rho\otimes\delta} || \rho \otimes \delta) = f_{\alpha,z}(f_1^{\rho} || \rho) f_{\alpha,z}(f_1^{\delta} || \delta).
		\end{align*}
		Given that $f_{\alpha,z}(f_1^{\rho} || \rho) \leqslant 1$ and $f_{\alpha,z}(f_1^{\delta} || \delta) \leqslant 1$, it follows that:
		\begin{align*}
			[1 - f_{\alpha,z}(f_1^{\rho} || \rho)] + [1 - f_{\alpha,z}(f_1^{\delta} || \delta)] &\geqslant 1 - f_{\alpha,z}(f_1^{\rho} || \rho) f_{\alpha,z}(f_1^{\delta} || \delta) \\
			&\geqslant [1 - f_{\alpha,z}(f_1^{\rho} || \rho)][1 - f_{\alpha,z}(f_1^{\delta} || \delta)].
		\end{align*}
		Therefore, we conclude:
		\begin{align*}
			\mathcal{T}^{\text{GR}}_{\alpha,z}(\rho) + \mathcal{T}^{\text{GR}}_{\alpha,z}(\delta) \geqslant \mathcal{T}^{\text{GR}}_{\alpha,z}(\rho \otimes \delta) \geqslant \mathcal{T}^{\text{GR}}_{\alpha,z}(\rho) \mathcal{T}^{\text{GR}}_{\alpha,z}(\delta).
		\end{align*}
	\end{proof}
	\par
	Since
	\begin{align*}
		\operatorname{Tr}(f_1^{\frac{\alpha}{2z}}\rho^{\frac{1-\alpha}{z}}f_1^{\frac{\alpha}{2z}})^z
		&= \operatorname{Tr}(f_1 \rho^{\frac{1-\alpha}{z}} f_1)^z \\
		&= \operatorname{Tr}(\bra{f_1}\rho^{\frac{1-\alpha}{z}}\ket{f_1} \cdot \ket{f_1}\bra{f_1})^z \\
		&= (\bra{f_1}\rho^{\frac{1-\alpha}{z}}\ket{f_1})^z \cdot \operatorname{Tr}(\ket{f_1}\bra{f_1}) \\
		&= (\bra{f_1}\rho^{\frac{1-\alpha}{z}}\ket{f_1})^z,
	\end{align*}
	the texture measure $\mathcal{T}^{\text{GR}}_{\alpha,z}$ defined via the $\alpha$-$z$ R\'enyi relative entropy can be equivalently written as:
	\begin{align*}
		\mathcal{T}^{\text{GR}}_{\alpha,z}(\rho) = 1 - (\bra{f_1}\rho^{\frac{1-\alpha}{z}}\ket{f_1})^z.
	\end{align*}
	\par
	This formulation allows $\mathcal{T}^{\text{GR}}_{\alpha,z}(\rho)$ to be connected with some known texture measures. For instance, choosing $\alpha = z = 0.5$, the expression $2\mathcal{T}^{\text{GR}}_{\alpha,z}(\rho)$ becomes equivalent to the texture measure based on the Bures distance \cite{wang2025quantifying}. Choosing $\alpha = 1-\mu$ and $z=1$, the expression $\frac{1}{1-\mu}\mathcal{T}^{\text{GR}}_{\alpha,z}(\rho)$ becomes equivalent to the texture measure based on Tsallis relative entropy \cite{zhang2025quantum}.
	
	\section{Relationships Between Texture Measures}
	\par\noindent\par 
	This section first reviews several established valid texture measures within the texture resource theory, and subsequently establishes mathematical relationships between them.
	\par
	The state rugosity texture measure, introduced by Parisio \cite{parisio2024quantum}, is defined as:
	\begin{align*}
		\mathcal{T}_{\text{SR}}(\rho)=-\ln(\bra{f_1}\rho\ket{f_1}).
	\end{align*}
	\par
	The texture measure based on the trace norm $\norm{A}_{\text{Tr}}=\operatorname{Tr}\sqrt{A^{\dagger}A}$ \cite{wang2025quantifying} is defined as:
	\begin{align*}
		\mathcal{T}_{\operatorname{Tr}}(\rho)=\frac{1}{2}\norm{\rho-f_1}_{\operatorname{Tr}}.
	\end{align*}
	\par
	The texture measure based on the fidelity $\mathcal{F}(\rho,f_1)=\operatorname{Tr}\sqrt{\sqrt{\rho}f_1\sqrt{\rho}}$ \cite{wang2025quantifying} is defined as:
	\begin{align*}
		\mathcal{T}_{\text{F}}(\rho)=1-\bra{f_1}\rho\ket{f_1}.
	\end{align*}
	\par
	The texture measure based on weight \cite{zhang2025quantum} is defined as:
	\begin{align*}
		\mathcal{T}_{\text{w}}(\rho)=\min_{s}\{s\geqslant0:\rho=(1-s)f_1+s\tau,\ \tau\in\mathcal{D}(\mathcal{H})\}.
	\end{align*}
	\par
	The texture measure based on the sandwiched R\'enyi relative entropy \cite{cao2026generalized} is defined as:
	\begin{align*}
		\mathcal{T}^{\text{R}}_{\alpha}(\rho)=\frac{1}{1-\alpha}\left[1- \big(\bra{f_1}\rho^{\frac{1-\alpha}{\alpha}}\ket{f_1}\big)^{\frac{\alpha}{1-\alpha}}\right] ,\quad \alpha\in[\frac{1}{2},1).
	\end{align*}
	\par
	Based on the above definitions, we establish the following series of inequalities to clarify the relative magnitudes and connections between different measures.
	\begin{proposition}
		For any density matrix $\rho$, $\mathcal{T}_{\text{F}}(\rho)\leqslant\mathcal{T}_{\text{SR}}(\rho)$ holds.
	\end{proposition}
	\begin{proof}
		From the definitions, we have $\mathcal{T}_{\text{SR}}(\rho)=-\ln(1-\mathcal{T}_{\text{F}}(\rho))$. Using the inequality:
		\begin{align*}
			-\ln(1-x)\geqslant x, \quad x\in[0,1),
		\end{align*}
		we obtain $\mathcal{T}_{\text{SR}}(\rho) \geqslant \mathcal{T}_{\text{F}}(\rho)$, concluding the proof.
	\end{proof}
	\begin{proposition}
		For any density matrix $\rho$, $1-\sqrt{1-\mathcal{T}_{\text{F}}(\rho)}\leqslant\mathcal{T}_{\operatorname{Tr}}(\rho)\leqslant\sqrt{\mathcal{T}_{\text{F}}(\rho)}$ holds.
	\end{proposition}
	\begin{proof}
		The Fuchs-van de Graaf inequality relate fidelity and trace distance \cite{zhang2016lower}:
		\begin{align*}
			1-\frac{1}{2}\norm{\rho-\sigma}_{\operatorname{Tr}}\leqslant\mathcal{F}(\rho,\sigma)\leqslant\sqrt{1-\frac{1}{4}\norm{\rho-\sigma}_{\operatorname{Tr}}^2}.
		\end{align*}
		The conclusion follows directly.
	\end{proof}
	\begin{proposition}
		For any density matrix $\rho$, $\mathcal{T}_{\text{F}}(\rho)\leqslant\mathcal{T}_{\text{w}}(\rho)$ holds.
	\end{proposition}
	\begin{proof}
		Choose $s_0\geqslant 0$ and $\tau_0\in\mathcal{D}(H)$ such that the decomposition $\rho=(1-s_0)f_1+s_0\tau_0$ holds. Then, we compute:
		\begin{align*}
			\bra{f_1}\rho\ket{f_1}
			&=\bra{f_1}(1-s_0)f_1+s_0\tau_0\ket{f_1} \\
			&=(1-s_0)\bra{f_1}f_1\ket{f_1}+s_0\bra{f_1}\tau_0\ket{f_1} \\
			&=1-s_0+s_0\bra{f_1}\tau_0\ket{f_1}.
		\end{align*}
		Since $\bra{f_1}\tau_0\ket{f_1}\geqslant 0$, we have $\bra{f_1}\rho\ket{f_1}\geqslant 1-s_0$, which implies:
		\begin{align*}
			\mathcal{T}_{\text{F}}(\rho)\leqslant s_0.
		\end{align*}
		This inequality holds for every $s$ in all feasible decompositions satisfying $\rho=(1-s)f_1+s\tau$. By the definition of $\mathcal{T}_{\text{w}}(\rho)$, we conclude:
		\begin{align*}
			\mathcal{T}_{\text{F}}(\rho)\leqslant\mathcal{T}_{\text{w}}(\rho).
		\end{align*}
	\end{proof}
	\begin{proposition}
		For any density matrix $\rho$, $\mathcal{T}_{\text{F}}(\rho)\leqslant(1-\alpha)\mathcal{T}^{\text{R}}_{\alpha}(\rho)$ holds.
	\end{proposition}
	\begin{proof}
		For $\alpha\in[\frac{1}{2},1)$, we have $\frac{\alpha}{1-\alpha}\geqslant 1$. Using Lemma 1 and the result from \cite{cao2026generalized}, we get:
		\begin{align*}
			\big(\bra{f_1}\rho^{\frac{1-\alpha}{\alpha}}\ket{f_1}\big)^{\frac{\alpha}{1-\alpha}}
			&= \operatorname{Tr}\big(f_1\rho^{\frac{1-\alpha}{\alpha}}f_1\big)^{\frac{\alpha}{1-\alpha}} \\
			&\leqslant \operatorname{Tr}\big(f_1^{\frac{\alpha}{1-\alpha}}\rho f_1^{\frac{\alpha}{1-\alpha}}\big) \\
			&= \operatorname{Tr}\big(f_1\rho f_1\big) = \bra{f_1}\rho\ket{f_1}.
		\end{align*}
		Thus, we have:
		\begin{align*}
			1-\big(\bra{f_1}\rho^{\frac{1-\alpha}{\alpha}}\ket{f_1}\big)^{\frac{\alpha}{1-\alpha}}\geqslant 1-\bra{f_1}\rho\ket{f_1}.
		\end{align*}
		Multiplying both sides by the positive factor $\frac{1}{1-\alpha}$ yields $(1-\alpha)\mathcal{T}^{\text{R}}_{\alpha}(\rho)\geqslant\mathcal{T}_{\text{F}}(\rho)$.
	\end{proof}
	\begin{proposition}
		For any density matrix $\rho$ and two parameters $\alpha_1\geqslant\alpha_2$, $\mathcal{T}^{\text{R}}_{\alpha_1}(\rho)\geqslant\mathcal{T}^{\text{R}}_{\alpha_2}(\rho)$ holds.
	\end{proposition}
	\begin{proof}
		When $\alpha_1\geqslant\alpha_2$, we have $\frac{\alpha_1(1-\alpha_2)}{\alpha_2(1-\alpha_1)}\geqslant 1$. Applying Lemma 1 yields:
		\begin{align*}
			\operatorname{Tr}\big(f_1\rho^{\frac{1-\alpha_2}{\alpha_2}}f_1\big)^{\frac{\alpha_2}{1-\alpha_2}}
			&= \operatorname{Tr}\big(f_1\rho^{\frac{1-\alpha_1}{\alpha_1}\cdot\frac{\alpha_1(1-\alpha_2)}{\alpha_2(1-\alpha_1)}}f_1\big)^{\frac{\alpha_2}{1-\alpha_2}} \\
			&\geqslant \operatorname{Tr}\big(f_1\rho^{\frac{1-\alpha_1}{\alpha_1}}f_1\big)^{\frac{\alpha_2}{1-\alpha_2}\cdot\frac{\alpha_1(1-\alpha_2)}{\alpha_2(1-\alpha_1)}} \\
			&= \operatorname{Tr}\big(f_1\rho^{\frac{1-\alpha_1}{\alpha_1}}f_1\big)^{\frac{\alpha_1}{1-\alpha_1}}.
		\end{align*}
		Therefore,
		\begin{align*}
			1-\big(\bra{f_1}\rho^{\frac{1-\alpha_2}{\alpha_2}}\ket{f_1}\big)^{\frac{\alpha_2}{1-\alpha_2}}
			\leqslant 1-\big(\bra{f_1}\rho^{\frac{1-\alpha_1}{\alpha_1}}\ket{f_1}\big)^{\frac{\alpha_1}{1-\alpha_1}}.
		\end{align*}
		Furthermore, noting that $\frac{1}{1-\alpha_1} \geqslant \frac{1}{1-\alpha_2}$ for $\alpha_1 \geqslant \alpha_2$, we conclude $\mathcal{T}^{\text{R}}_{\alpha_1}(\rho)\geqslant\mathcal{T}^{\text{R}}_{\alpha_2}(\rho)$.
	\end{proof}
	
	\section{Texture Witnesses}
	\par\noindent\par  
	Beyond quantification, the experimental detection of texture resources represents a promising direction. In quantum resource theory, a witness is an effective tool for experimentally detecting and verifying specific quantum resources, widely applied in areas such as entanglement \cite{Horodecki1996separability,Horodecki2009Quantum}, coherence \cite{ma2021detecting}, block coherence \cite{chen2023detecting}, and imaginarity \cite{zhang2025onim}. For instance, in coherence resource theory, coherence witnesses can determine whether a state possesses coherence. Inspired by this, we can define texture witnesses within the texture resource theory to probe whether a quantum state possesses texture as a resource.
	\begin{definition}
		A Hermitian operator $W$ is called a \textbf{texture witness} if and only if it satisfies:
		\begin{eqnarray*}
			\operatorname{Tr}(W f_1) \geq 0.
		\end{eqnarray*}
		If a texture witness $W$ further satisfies that there exists at least one texture state $\rho$ such that $\operatorname{Tr}(W\rho) < 0$, then it is called a \textbf{strict texture witness}.
	\end{definition}
	The broader condition $\operatorname{Tr}(W f_1) \geq 0$ defines the complete set of Hermitian operators that are, at minimum, do not yield a false positive detection on the free state (they never falsely indicate texture in $f_1$). Within this set, the strict texture witnesses are the operational subset capable of conclusively detecting the resource. By this definition, there exist trivial texture witnesses. The free state $f_1$ itself, viewed as an operator, is a texture witness since $\operatorname{Tr}(f_1 f_1) = 1 \geq 0$. However, for any quantum state $\rho$, we have $\operatorname{Tr}(f_1 \rho) = \langle f_1|\rho|f_1\rangle \geq 0$. Therefore, $W = f_1$ is not a strict texture witness.
	\par 
	The above definition implies that if measuring a strict texture witness $W$ on a state $\rho$ yields a negative expectation value ($\operatorname{Tr}(W\rho) < 0$), one can conclude that $\rho$ necessarily contains texture resources. Conversely, a non-negative result is inconclusive (the state could be free or a texture state not detected by this particular witness). 
	\par
	Analogous to the general construction method for coherence witnesses, a universal construction also exists for texture witnesses:
	\begin{theorem}
		For any Hermitian operator $A$, the operator
		\begin{align*}
			W = \Delta_T(A) - A = \bra{f_1} A \ket{f_1} f_1 - A
		\end{align*}
		is a texture witness, where $\Delta_T$ is the detexturing operation: $\Delta_T(X) = \operatorname{Tr}(X f_1) f_1$.
		\par 
		Furthermore, if $A$ satisfies $\bra{f_1} A \ket{f_1} f_1 \neq A$, then the constructed $W$ is a strict texture witness.
	\end{theorem}
	
	\begin{proof}
		First, we verify that $W$ satisfies Condition 1 of Definition 1. Compute $\operatorname{Tr}(W f_1)$:
		\begin{align*}
			\operatorname{Tr}(W f_1)
			&= \operatorname{Tr}\bigl( \bra{f_1} A \ket{f_1} f_1^2 - A f_1 \bigr) \\
			&= \bra{f_1} A \ket{f_1} \operatorname{Tr}(f_1) - \bra{f_1} A \ket{f_1} \\
			&= \bra{f_1} A \ket{f_1} - \bra{f_1} A \ket{f_1} \\
			&= 0 \geqslant 0.
		\end{align*}
		Thus, for any Hermitian operator $A$, $W$ is a texture witness.
		\par 
		Next, we prove the strictness condition. Assume $\bra{f_1} A \ket{f_1} f_1 \neq A$. Define the deviation operator
		\begin{align*}
			B = A - \bra{f_1} A \ket{f_1} f_1,
		\end{align*}
		so that $W = -B$. Note that $B \neq 0$ is Hermitian and satisfies $\bra{f_1} B \ket{f_1} = 0$.
		\par
		Consider the spectral decomposition of $B$:
		\begin{align*}
			B = \sum_i \lambda_i \ket{\psi_i}\bra{\psi_i},
		\end{align*}
		where $\lambda_i$ are real eigenvalues and $\{\ket{\psi_i}\}$ forms an orthonormal basis. The condition $\bra{f_1} B \ket{f_1} = 0$ implies
		\begin{align*}
			\sum_i \lambda_i |\bra{f_1}\ket{\psi_i}|^2= 0.
		\end{align*}
		Since $B \neq 0$, not all eigenvalues are zero, and the above equality forces $B$ to have at least one positive eigenvalue. Let $\lambda_{\max} > 0$ be a positive eigenvalue of $B$ with corresponding eigenstate $\ket{\psi_{\max}}$. Now take the texture state $\rho = \ket{\psi_{\max}}\bra{\psi_{\max}}$. Its expectation value is
		\begin{align*}
			\operatorname{Tr}(W\rho) = -\bra{\psi_{\max}} B \ket{\psi_{\max}} = -\lambda_{\max} < 0.
		\end{align*}
		Moreover, $\rho$ is indeed a texture state because if $\ket{\psi_{\max}} = e^{i\phi}\ket{f_1}$, then
		\begin{align*}
			\lambda_{\max} = \bra{f_1} B \ket{f_1} = 0,
		\end{align*}
		contradicting $\lambda_{\max} > 0$. Therefore, $W$ is a strict texture witness.
	\end{proof}
	The universal construction above provides a template for generating (strict) texture witnesses. By choosing different Hermitian operators $A$, we can obtain witnesses with different detection characteristics. Here is an example of a texture witness:
	\begin{example} Strict texture witness $W_1=f_1-I$.
		\par
		Let $A_1 = I$, and substitute it into the universal construction formula:
		\begin{align*}
			W_1 = \bra{f_1} I \ket{f_1} f_1 - I = 1 \cdot f_1 - I = f_1 - I.
		\end{align*}
		Then $W_1 = f_1 - I$ is a texture witness. For the free state $f_1$, we have:
		\begin{align*}
			\operatorname{Tr}(W_1 f_1) = \bra{f_1} f_1 \ket{f_1} - \operatorname{Tr}(f_1) = 1 - 1 = 0.
		\end{align*}
		For any texture state $\rho$, since:
		\begin{align*}
			\operatorname{Tr}(W_1 \rho)=\operatorname{Tr}(f_1\rho) - \operatorname{Tr}(\rho) = \bra{f_1}\rho\ket{f_1} - 1 = - \mathcal{T}_{\text{F}}(\rho),
		\end{align*}
		where $\mathcal{T}_{\text{F}}(\rho)$ is the fidelity-based measure. Therefore, any texture state $\rho$ always satisfies $\operatorname{Tr}(W_1 \rho) < 0$. This means the texture witness $W_1$ can detect all texture states, and its negative expectation value directly gives the fidelity measure of texture.
	\end{example}
	\par
	It is important to note that the set of texture witnesses is broader than the set generated by the universal construction; the universal construction yields only a proper subset. We can present texture witnesses that are not constructed via Theorem 2. To this end, we first introduce a generator $G$ associated with the texture-free state $f_1$:
	\begin{definition}
		Define the generator $G$ as:
		\begin{align*}
			G = 2\ket{f_1}\bra{f_1} - I = 2f_1 - I.
		\end{align*}
		This generator satisfies $G^2 = I$, with eigenvalues $+1$ (corresponding to eigenstate $\ket{f_1}$) and $-1$ (corresponding to all states orthogonal to $\ket{f_1}$).
	\end{definition}
	This leads to the following example:
	\begin{example} Strict texture witness $W_{\theta} = \cos\theta\cdot I+\sin\theta\cdot G, \theta\in\left( \frac{\pi}{4},\frac{3\pi}{4}\right] $.
		\par
		Let $W_{\theta}$ denote the family of operators $\cos\theta\cdot I+\sin\theta\cdot G$, that is, $W_{\theta}= (\cos\theta-\sin\theta)I+2\sin(\theta)f_1$. For $W_{\theta}$ to be a valid texture witness, it must satisfy the two conditions in the definition. For Condition 1, we have:
		\begin{align*}
			\operatorname{Tr}(W_{\theta} f_1) = (\cos\theta-\sin\theta)\operatorname{Tr}(f_1)+2\sin(\theta)\operatorname{Tr}(f_1)=\cos\theta+\sin\theta.
		\end{align*}
		Thus, $\theta$ must satisfy $\cos\theta+\sin\theta\geqslant 0$, yielding $\theta\in\left[0,\frac{3\pi}{4}\right]\cup\left[ \frac{7\pi}{4},2\pi\right]$.
		\par
		For Condition 2, there must exist at least one texture state $\rho$ with $\bra{f_1}\rho\ket{f_1}\in[0,1]$ such that $\operatorname{Tr}(W_{\theta}\rho)<0$. Let $x=\bra{f_1}\rho\ket{f_1}$, then we require:
		\begin{align*}
			g(x)=\cos\theta-\sin\theta+2\sin\theta\cdot x<0.
		\end{align*}
		When $\sin\theta\geqslant 0$, $g(x)$ is monotonically increasing. To satisfy existence, it suffices that $g(0)=\cos\theta-\sin\theta<0$, giving $\theta\in\left( \frac{\pi}{4},\pi\right]$. When $\sin\theta<0$, $g(x)$ is monotonically decreasing. To satisfy existence, it suffices that $g(1)=\cos\theta+\sin\theta<0$, giving $\theta\in\left( \pi,\frac{7\pi}{4}\right)$, which contradicts Condition 1.
		\par
		In summary, $W_{\theta}$ is a valid strict texture witness when $\theta\in\left( \frac{\pi}{4},\frac{3\pi}{4}\right]$.
		\par
		Furthermore, $W_{\theta}$ has a direct relationship with the fidelity-based texture measure $\mathcal{T}_{\text{F}}(\rho)$:
		\begin{align*}
			\mathcal{T}_{\text{F}}(\rho)=\frac{\cos\theta+\sin\theta-\operatorname{Tr}(W_{\theta}\rho)}{2\sin\theta}.
		\end{align*}
		This represents a selective detection scheme. The texture witness $W_{\theta}$ yields a negative result if and only if the fidelity-based texture measure exceeds a specific threshold:
		\begin{align*}
			\operatorname{Tr}(W_{\theta}\rho)<0
			\Leftrightarrow
			\mathcal{T}_{\text{F}}(\rho)>\frac{\cos\theta+\sin\theta}{2\sin\theta}.
		\end{align*}
		The parameter $\theta$ adjusts the value of this threshold. By tuning $\theta$, we can alter the sensitivity of the witness, adapting it to different detection requirements. For instance, choosing $\theta=\frac{\pi}{2}$ gives $W_{\pi/2} = G = 2f_1 - I$. In this case, the texture witness is sensitive only to states $\rho$ whose squared fidelity with the free state $f_1$ is below 50\% (i.e., $[\mathcal{F}(\rho,f_1)]^2=\big(\operatorname{Tr}\sqrt{\sqrt{\rho}f_1\sqrt{\rho}}\big)^2<0.5$).
		\par
		We also point out that there exists no Hermitian operator $A$ that can generate $W_{\pi/2}$ via the universal construction formula $W = \Delta_T(A) - A$. Suppose such an $A_{\pi/2}$ existed, satisfying:
		\begin{align*}
			\bra{f_1} A_{\pi/2} \ket{f_1} f_1 - A_{\pi/2} = 2f_1 - I,
		\end{align*}
		which implies:
		\begin{align*}
			A_{\pi/2}=\bra{f_1} A_{\pi/2} \ket{f_1} f_1 - 2f_1 + I.
		\end{align*}
		Taking the expectation value with $f_1$ then yields:
		\begin{align*}
			\bra{f_1} A_{\pi/2} \ket{f_1} = \big(\bra{f_1} A_{\pi/2} \ket{f_1}-2\big) + 1 = \bra{f_1} A_{\pi/2} \ket{f_1} - 1,
		\end{align*}
		a contradiction. Therefore, no Hermitian operator $A$ can generate $W_{\pi/2}$, demonstrating that the set of texture witnesses is indeed broader than the set from the universal construction.
	\end{example}
	Similarly, we can directly construct the following texture witness without using Theorem 1:
	\begin{example} Strict texture witness $W^{jk}_{\varphi}$.
		\par
		Take any pair of distinct basis states $\ket{j}$ and $\ket{k}$ ($j \neq k$) from the computational basis $\{\ket{j}\}_{j=1}^d$, along with an arbitrary phase parameter $\varphi\in(0,2\pi)$. Construct the following Hermitian operator:
		\begin{align*}
			W^{jk}_{\varphi} = \frac{2\cos\varphi}{d} I - \left( e^{i\varphi} \ket{j}\bra{k} + e^{-i\varphi} \ket{k}\bra{j} \right).
		\end{align*}
		For $f_1$, the expectation value of $W^{jk}_{\varphi}$ is:
		\begin{align*}
			\operatorname{Tr}(W^{jk}_{\varphi} f_1)
			&= \operatorname{Tr}\left( \frac{2\cos\varphi}{d} I f_1 \right) - e^{i\varphi}\operatorname{Tr}(\ket{j}\bra{k} f_1) - e^{-i\varphi}\operatorname{Tr}(\ket{k}\bra{j} f_1) \\
			&= \frac{2\cos\varphi}{d} - e^{i\varphi} \bra{k} f_1 \ket{j} - e^{-i\varphi} \bra{j} f_1 \ket{k} \\
			&= \frac{2\cos\varphi}{d} - e^{i\varphi} \cdot \frac{1}{d} - e^{-i\varphi} \cdot \frac{1}{d} \\
			&= \frac{1}{d} \big[2\cos\varphi - (e^{i\varphi} + e^{-i\varphi})\big] \\
			&= 0 \ge 0.
		\end{align*}
		This satisfies Condition 1.
		\par
		Consider the pure state $\rho=\ket{\psi}\bra{\psi}$, where $\ket{\psi} = \frac{1}{\sqrt{2}}\big(\ket{j} + e^{-i\varphi} \ket{k}\big)$. Its off-diagonal element is $\rho_{jk} = \bra{j}\rho\ket{k} = \frac{1}{2}e^{i\varphi}$. For $\varphi\in(0,2\pi)$, the expectation value of the witness for this state is:
		\begin{align*}
			\operatorname{Tr}(W^{jk}_{\varphi}\rho)
			&= \operatorname{Tr}\left( \frac{2\cos\varphi}{d} I \rho \right) - e^{i\varphi}\operatorname{Tr}(\ket{j}\bra{k} \rho) - e^{-i\varphi}\operatorname{Tr}(\ket{k}\bra{j} \rho) \\
			&= \frac{2\cos\varphi}{d} - e^{i\varphi} \bra{k} \rho \ket{j} - e^{-i\varphi} \bra{j}\rho\ket{k} \\
			&= \frac{2\cos\varphi}{d} - e^{i\varphi} \rho_{jk}^{\dagger} - e^{-i\varphi}\rho_{jk} \\
			&= \frac{2\cos\varphi}{d} - 2\Re{e^{-i\varphi}\rho_{jk}} \\
			&= \frac{2\cos\varphi}{d}-1<0,\quad (\text{Note that: }d>1).
		\end{align*}
		In conclusion, $W^{jk}_{\varphi}$ is a strict texture witness.
		\par
		Particularly, in the imaginarity resource theory, the matrix elements of an imaginarity-free state are all real. The above texture witness can be used to detect imaginarity resources. Let:
		\begin{align*}
			W^{jk}_{I+}=W^{jk}_{\pi/2},\quad W^{jk}_{I-}=W^{jk}_{3\pi/2}.
		\end{align*}
		Then, for any quantum state $\sigma$, their respective expectation values are:
		\begin{align*}
			\operatorname{Tr}(W^{jk}_{I+}\sigma)&=- e^{i{\frac{\pi}{2}}} \bra{k} \sigma \ket{j} - e^{-i{\frac{\pi}{2}}} \bra{j}\sigma\ket{k} =-2\Im{\sigma_{jk}},\\
			\operatorname{Tr}(W^{jk}_{I-}\sigma)&=- e^{i{\frac{3\pi}{2}}} \bra{k} \sigma \ket{j} - e^{-i{\frac{3\pi}{2}}} \bra{j}\sigma\ket{k} =2\Im{\sigma_{jk}}.
		\end{align*}
		Thus, we have:
		\begin{align*}
			\operatorname{Tr}(W^{jk}_{I+}\sigma)<0 &\Leftrightarrow \Im{\sigma_{jk}}>0,\\
			\operatorname{Tr}(W^{jk}_{I-}\sigma)<0 &\Leftrightarrow \Im{\sigma_{jk}}<0,
		\end{align*}
		showing that the texture witness $W^{jk}_{I+}$ specifically detects positive imaginary parts, while $W^{jk}_{I-}$ detects negative imaginary parts.
	\end{example}
	In summary, texture witnesses provide a feasible scheme for the experimental detection of texture resources. By measuring the expectation value of specific Hermitian operators, we can determine whether a quantum state possesses texture. Moreover, the measurement results of certain witnesses (e.g., $W_1$) can directly yield a quantitative measure of texture. This enriches the toolkit of texture resource theory, incorporating not only mathematical measures but also experimentally operable detection methods.
	
	\section{Conclusion}
	\par\noindent\par 
	In this work, we have conducted research within the framework of the quantum-state texture resource theory, with our main contributions lying in two aspects: the study of texture measures and the establishment of texture witnesses. In terms of quantification, we constructed a new texture measure $\mathcal{T}^{\text{GR}}_{\alpha,z}(\rho)$ based on the $\alpha$-$z$ R\'enyi relative entropy and established mathematical relationships between various existing measures, providing new perspectives and bounds for the quantitative description of texture. Regarding detection, we introduced, for the first time, the framework of texture witnesses. We provided a universal construction method and offered several experimentally measurable witness instances with distinct characteristics, which may furnish a feasible scheme for the experimental identification and verification of texture resources.
	\par
	Future research could delve deeper in the following directions: (1) exploring applications of the texture resource in specific quantum information processing tasks; (2) investigating the transformation rules of texture in composite tensor systems; (3) designing more efficient texture witness schemes with regard to their detection capability and experimental feasibility; and (4) constructing other qualified texture measures that may offer advantages in calculability or measurability.
	
	\section*{Data availability statement}
	No new data were created or analyzed in this study.
	
	\section*{Acknowledgments}
	This work was supported by National Natural Science Foundation of China (Grants No. 12271474).

\end{document}